\documentclass[oneside,english]{amsart}
\usepackage[T1]{fontenc}
\usepackage[latin9]{inputenc}
\usepackage{babel}
\usepackage{bm}
\usepackage{amstext}
\usepackage{amsthm}
\usepackage{amssymb}
\usepackage{graphicx}
\usepackage[unicode=true,pdfusetitle,
 bookmarks=true,bookmarksnumbered=false,bookmarksopen=false,
 breaklinks=false,pdfborder={0 0 1},backref=false,colorlinks=false]
 {hyperref}

\makeatletter
\numberwithin{equation}{section}
\numberwithin{figure}{section}
\theoremstyle{plain}
\newtheorem{thm}{\protect\theoremname}
\theoremstyle{definition}
\newtheorem{example}[thm]{\protect\examplename}
\theoremstyle{remark}
\newtheorem{rem}[thm]{\protect\remarkname}
\theoremstyle{definition}
\newtheorem{problem}[thm]{\protect\problemname}
\theoremstyle{plain}
\newtheorem{lem}[thm]{\protect\lemmaname}

\@ifundefined{showcaptionsetup}{}{%
 \PassOptionsToPackage{caption=false}{subfig}}
\usepackage{subfig}
\makeatother

\providecommand{\examplename}{Example}
\providecommand{\lemmaname}{Lemma}
\providecommand{\problemname}{Problem}
\providecommand{\remarkname}{Remark}
\providecommand{\theoremname}{Theorem}

\begin{document}
\title[Coordinate free effective diffusion in channels ]{ A coordinate free formulation of effective diffusion on channels}
\author{Carlos Valero Valdés\\
Departamento de Matemáticas\\
Universidad de Guanajuato\\
Guanajuato, México}
\begin{abstract}
We study diffusion processes in regions generated by ``sliding''
a cross section by the phase flow of vector filed on curved spaces
of arbitrary dimension. We do this by studying the effective diffusion
coefficient $D$ that arises when trying to reduce the $n$-dimensional
diffusion equation to a 1-dimensional diffusion equation by means
of a projection method. We use the mathematical language of exterior
calculus to derive a coordinate free formula for this coefficient
in both infinite and finite transversal diffusion rate cases. The
use of these techniques leads to a formula for $D$ which provides
a deeper understanding of effective diffusion than when using a coordinate
dependent approach.
\end{abstract}

\maketitle
\global\long\def\CC{\mathbb{C}}%

\global\long\def\RR{\mathbb{R}}%

\global\long\def\ZZ{\mathbb{Z}}%

\global\long\def\pd#1#2{\frac{\partial#1}{\partial#2}}%

\global\long\def\pdd#1#2{\frac{\partial^{2}#1}{\partial#2^{2}}}%

\global\long\def\mpd#1#2#3{\frac{\partial#1}{\partial#2\partial#3}}%

\global\long\def\pdo#1{\frac{\partial}{\partial#1}}%

\global\long\def\pddo#1{\frac{\partial^{2}}{\partial#1^{2}}}%

\global\long\def\der#1#2{\frac{d#1}{d#2}}%

\global\long\def\dder#1#2{\frac{d^{2}#1}{d#2^{2}}}%

\global\long\def\diff{d}%

\global\long\def\do#1{\frac{d}{d#1}}%

\global\long\def\map{\rightarrow}%

\global\long\def\tangent{T}%

\global\long\def\cotangent{\tangent^{*}}%

\global\long\def\SS{\mathcal{S}}%

\global\long\def\KK{\mathcal{K}}%

\global\long\def\NN{\mathcal{N}}%

\global\long\def\sym#1{\hbox{S}^{2}#1}%

\global\long\def\symz#1{\hbox{S}_{0}^{2}#1}%

\global\long\def\SO{\hbox{SO}}%

\global\long\def\GL{\hbox{GL}}%

\global\long\def\U{\hbox{U}}%

\global\long\def\tr{\hbox{tr}}%

\global\long\def\and{\text{ and }}%

\global\long\def\dv{\nabla\cdot}%

\global\long\def\JJ{\boldsymbol{J}}%

\global\long\def\DD{\boldsymbol{D}}%

\global\long\def\pp{\boldsymbol{p}}%

\global\long\def\vol{\hbox{vol}}%

\global\long\def\where{\text{ where }}%

\global\long\def\dxy{dx_{1}\wedge dx_{2}}%

\global\long\def\dxz{dx_{1}\wedge dx_{3}}%

\global\long\def\dyz{dx_{2}\wedge dx_{3}}%

\global\long\def\dxyz{dx_{1}\wedge dx_{2}\wedge dx_{3}}%

\global\long\def\for{\text{ for }}%

\global\long\def\LL{\mathcal{L}}%

\global\long\def\JJ{\mathcal{J}}%

\global\long\def\CH{\mathcal{C}}%

\global\long\def\WA{\mathcal{W}}%

\global\long\def\SE{\mathcal{S}}%

\global\long\def\fl#1{\mathcal{F}_{#1}}%

\global\long\def\vf{\omega_{vol}}%

\global\long\def\derz#1{\left.\frac{d}{d#1}\right|_{#1=0}}%

\global\long\def\pdoat#1{\left.\frac{\partial}{\partial#1}\right|}%

\global\long\def\AA{\mathcal{A}}%

\global\long\def\BB{\mathcal{B}}%

\global\long\def\el{\text{ is in }}%

\global\long\def\UU{\bm{U}}%

\global\long\def\VV{\bm{V}}%

\global\long\def\HH{\bm{H}}%

\global\long\def\GG{\mathcal{G}}%

\global\long\def\si{\text{ if }}%

\global\long\def\nn{\bm{n}}%

\global\long\def\aa{\bm{\alpha}}%

\global\long\def\tgt{\bm{t}}%

\global\long\def\vv{\bm{v}}%

\global\long\def\ww{\bm{w}}%

\global\long\def\bb{\bm{b}}%

\global\long\def\WW{\bm{W}}%

\global\long\def\dc{\mathcal{D}}%

\global\long\def\UUU{\mathcal{U}}%

\section{Introduction}

The purpose of this paper is to present a coordinate free formulation
of the theory of effective diffusion on channels. This problem has
been studied extensively in the literature with the use of specific
coordinate systems (e.g \cite{kn:aproximations,kn:bradley,kn:entropybarrierzwanzig,kn:kp-diffusion-projection,kn:kp-Extendex-fj-variational,kn:kp-fick-jacob-correction,kn:ogawa,kn:reguerarubi,kn:yariv,va:fj3dcurves,va:projdiffplanecurve}).
We tackle the coordinate free formulation by using modern tools from
differential geometry; more specifically: vector field flows and exterior
calculus. The advantages of taking this point of view is that it provides
a unified theory with the following properties.
\begin{enumerate}
\item The formulas obtained hold for channels of any dimension in arbitrary
flat and curved spaces.
\item By using this geometric approach one can gain a deeper and more intuitive
understanding of the formulas for the effective diffusion coefficient:
both in the finite and infinite transversal diffusion rate case. 
\item Our approach has lead us to identify the Fick-Jacobs equation as a
standard diffusion equation. To do this we need to change the metric
in the variable parametrizing the cross sections of the channel,  which
has also lead us to modify the the definition of effective density
function and effective diffusion coefficient used in most of the literature.
\item Our formulas hold for an arbitrary selection of cross sections of
the channel. This generality has lead us to identify the concepts
of natural and imposed projection maps.
\end{enumerate}

\subsection{Plan of the paper}

In section 2 we show how to generate channels using vector fields
on arbitrary spaces, and how these provide us with cross sections
which allows us to reduce a general diffusion equation to a diffusion
equation with only one spatial variable. In section 3 we present formulas
for the effective diffusion coefficient (both in the infinite and
finite transversal diffusion rate cases) avoiding the use of the language
exterior calculus, so that the main results can be understood in a
non-technical manner. In fact, the main concept needed in our formulas
is simply that of the flux of a vector field across a hyper-surface.
We show that there is a special choice of cross section of a channel
in which the formulas for the effective diffusion coefficient for
the finite and infinite transversal rate cases coincide. Section 4
contains the coordinate free derivation of our formulas using exterior
calculus, and we also how we can recover the coordinate dependent
formulas from our general result.

\section{Channel geometry through vector field flows}

\begin{figure}
\includegraphics[scale=0.25]{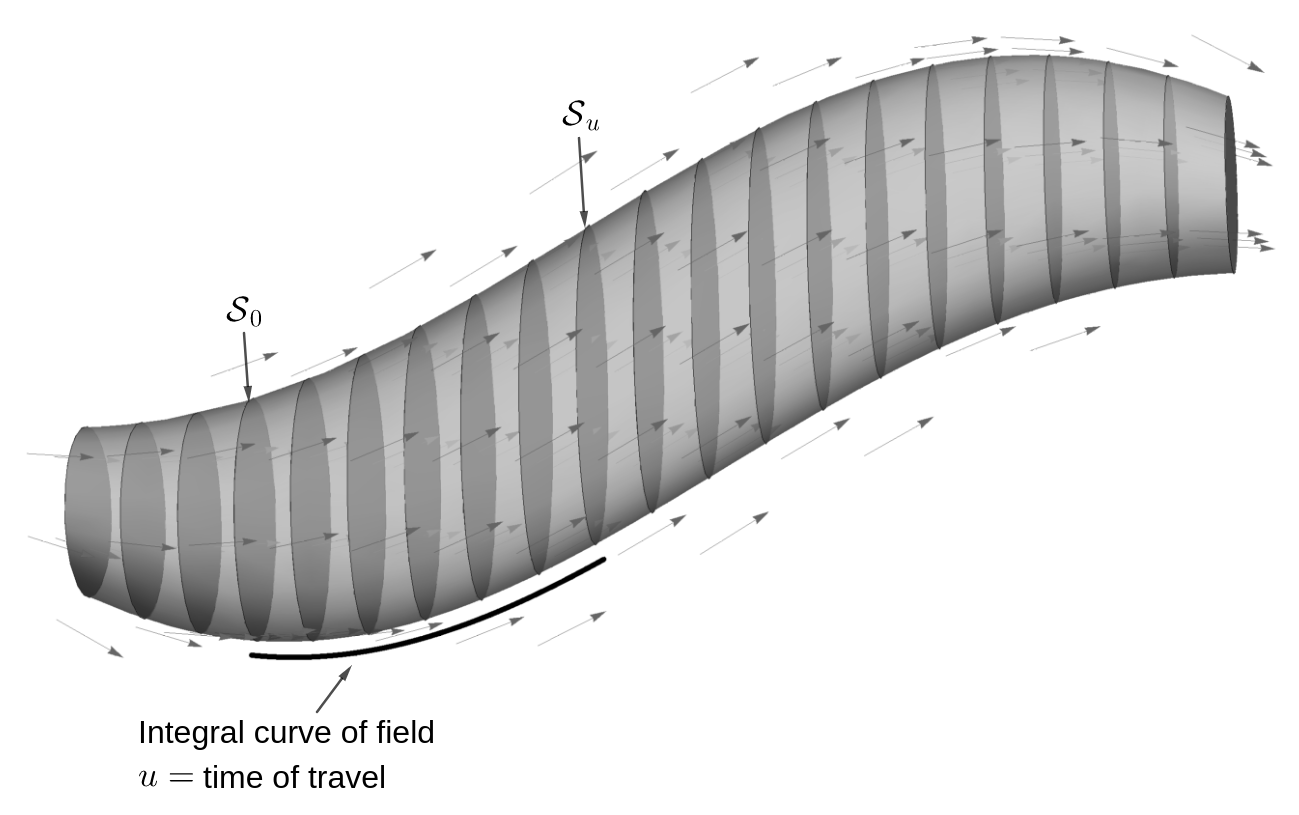}\caption{\label{fig:Parametric-channel}A channel generated by ``sliding''
a cross section $\protect\SE_{0}$ along the integral curves of a
vector field. The wall $\protect\WA$ consists of the surface that
encapsulates the channel, but not including the cross sections at
its extreme sides.}
\end{figure}

We are interested in studying channel-like objects, which we will
denote by $\CH$, in an $n-$dimensional space $M$\footnote{Usually $M$ stands for flat space of dimension either two or three,
but our results hold for the general case where $M$ is an arbitrary
$n$-dimensional oriented Riemannian manifold.}. We will construct $\CH$ using the following procedure (see Figure
\ref{fig:Parametric-channel}). Let $\SE_{0}$ be a $(n-1)$-dimensional
hyper-surface with boundary in $M$ and let $\UU$ a vector field
in $\CH$. For a given real number $u$, let $\SS_{u}$ be the hyper-surface
obtained by ``sliding'' $\SS_{0}$ along the \emph{integral curves}\footnote{An integral curve of $\UU$ is a curve $x=x(t)$ in $M$ that satisfies
$\der xt(t)=\UU(x(t))$.} of $\UU$ for a duration of $u$. We will refer to $\SE_{u}$ as
the cross section of $\CH$ at $u$. If $\CH$ is the union of the
cross sections $\SS_{u}$, we will say that the vector field $\UU$
generates $\CH$. If $\partial\SE_{u}$ is the boundary of $\SE_{u}$,
the wall $\WA$ of $\CH$ is the union of the sets $\partial\SE_{u}$
for all $u$'s. For $x$ in $\CH$ we will let $u(x)$ be equal to
the time it takes for a point in $\SE_{0}$ to reach $x$ (by following
an integral curve of $\UU$). In this context, we will refer to $u$
as a \emph{projection function}\footnote{Depending on the context will think $u$ as a scalar or as a function.}
for the channel. Notice that $\SS_{s}$ can be characterized as the
set of points in $\CH$ at which $u(x)=s$. 

As a particular case of the above construction consider a \emph{parametric
channel,} obtained by using a parametrization function $x=x(u,\text{\ensuremath{v)}}$
where $u$ is a scalar and $v$ belongs to some region in $(n-1)$-dimensional
space. In this case the generating vector field of the channel is
\[
\UU(x)=\pd xu(u(x),v(x)),
\]
where $u(x)$ and $v(x)$ are the $u$ and $v$ coordinates of the
point $x$ in $M$. 
\begin{figure}

\includegraphics[scale=0.95]{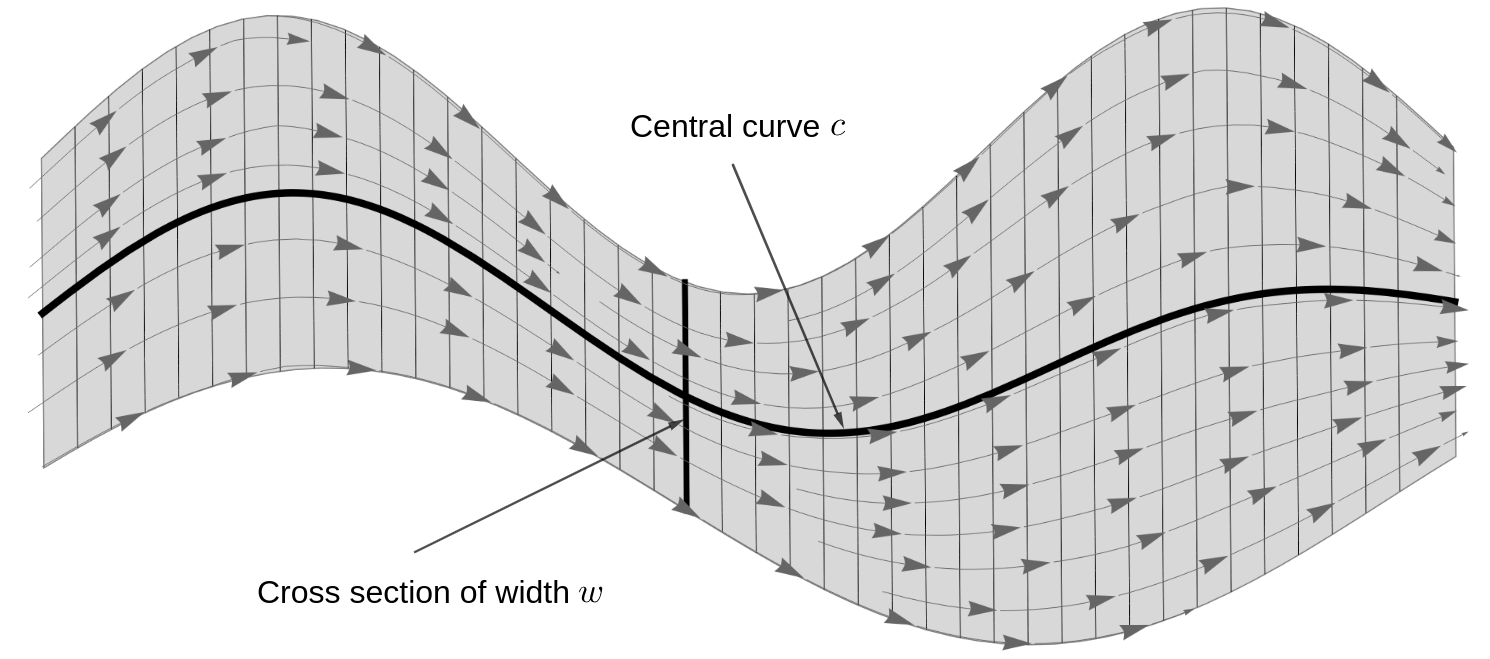}\caption{\label{fig:Parametric-Graph}Parametric channel with central function
$c=c(u)$ and width function $w=w(u)$.}
\end{figure}

\begin{example}
For $n=2$ let the variable $v$ be in re region $-1/2\leq v\leq1/2$
and define (see Figure \ref{fig:Parametric-Graph})
\[
x(u,v)=(u,c(u)+vw(u)),
\]
for scalar valued functions $c=c(u)$ and $w=w(u)$. We have that
\[
\pd xu(u,v)=(1,c'(u)+vw'(u)).
\]
If we write $x=(x_{1},x_{2})$, then from the formulas
\[
x_{1}=u\and x_{2}=c(u)+vw(u)
\]
we obtain
\[
v=\frac{x_{2}-c(x_{1})}{w(x_{1})}.
\]
Hence, the generating vector field of the channel is
\[
\UU(x_{1},x_{2})=\left(1,c'(x_{1})+\left(\frac{x_{2}-c(x_{1})}{w(x_{1})}\right)w'(x_{1})\right).
\]
A projection function for this field is 
\[
u(x_{1},x_{2})=x_{1},
\]
and the cross section $\SS_{u}$ is a line parallel to the $x_{2}$-axis
intersecting the $x_{1}$-axis at $(u,0)$.
\end{example}

\begin{rem}
We constructed the projection function $u$ in terms the vector field
$\UU$, by letting $u(x)$ be the time it takes for an integral curve
of $\UU$ starting at $\SS_{0}$ to reach $x$. Alternatively, we
could first select a scalar valued function $u$ in $\CH$ and then
construct a generating vector field $\UU$ in $M$ that satisfies
\[
\nabla u(x)\cdot\UU(x)=1\text{ for all x in }\CH.
\]
This condition implies that $u$ is a projection function of for $\UU$.
The initial cross section $\SS_{0}$ is then chosen so that $u(x)=0$
for all $x$ in $\SS_{0}$.
\end{rem}

\subsection{\label{subsec:Harmonic-projection-maps} Natural projection functions
and fields}

\begin{figure}
\subfloat[]{\includegraphics[scale=0.25]{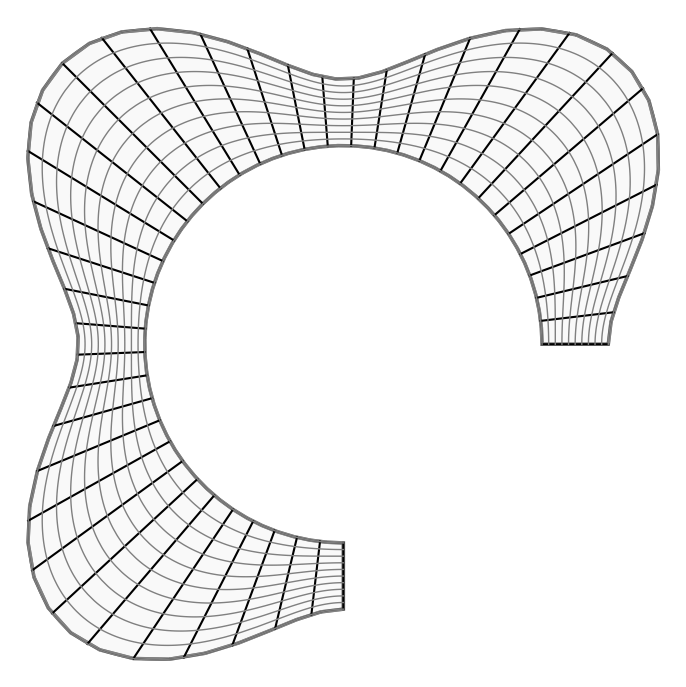}}\subfloat[]{\includegraphics[scale=0.25]{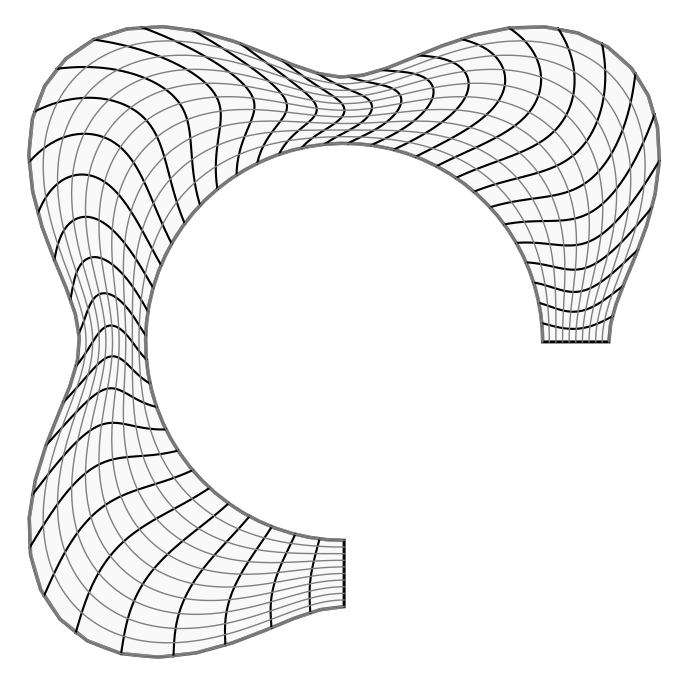}}

\caption{\label{fig:Two-cross-sections-2D}Two sets of cross sections associated
with two different generating vector fields of the same channel. }
\end{figure}

A channel $\CH$ can have many generating fields, which in general
produce different sets of cross sections (see Figure \ref{fig:Two-cross-sections-2D}
). This observation leads to the following problem.
\begin{problem}
Consider a channel $\CH$ with two fixed cross sections $\AA$ and
$\BB$. Is there a way to chose a generating vector field $\UU$ for
$\CH$ such that $\AA$ and $\BB$ are cross sections of $\CH$ generated
by $\UU$, and generated cross in between them ``fit'' the geometry
of $\CH$ in a ``natural way''?
\begin{figure}
\subfloat[\label{fig:Harmonic-Cross-sections}Cross sections]{\includegraphics[scale=0.5]{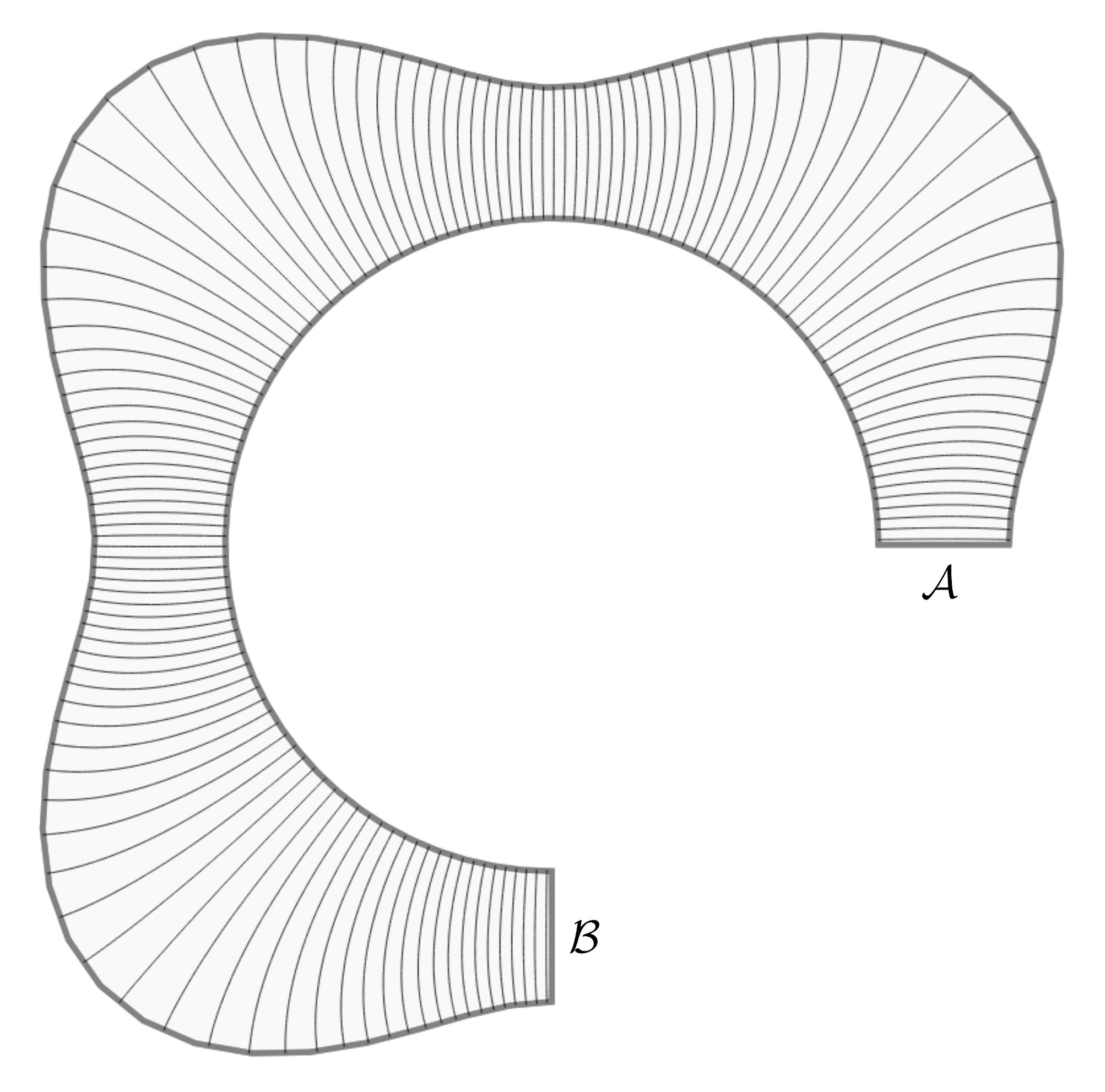}}\subfloat[Generating vector field]{\includegraphics[scale=0.5]{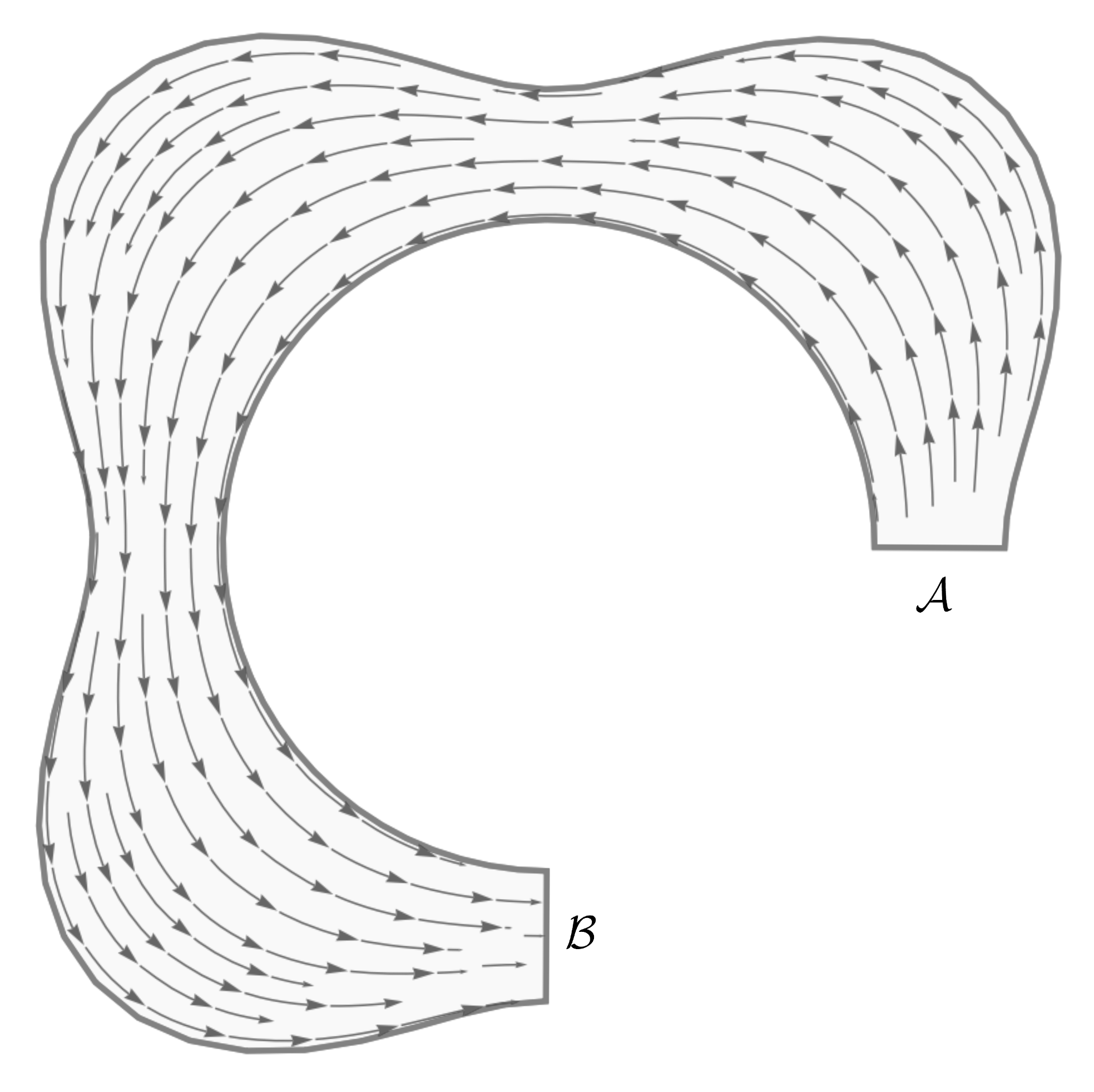}}\caption{\label{fig:Cross-section1}Cross sections and generating field associated
to a harmonic function $h$ on a channel. This function takes the
values $0$ and $1$ on $\protect\AA$ and $\protect\BB$ (respectively),
and its gradient has no flux across the other walls of the channel
(i.e in $\protect\WA$).}
\end{figure}
\end{problem}

We will argue that a ``natural way'' to choose $\UU$ is as follows.
For two different scalars $a$ and $b$ let $h$ be a harmonic function
(i.e $\Delta h=0)$ on $\CH$ such that $\nabla h$ has no flux across
$\WA$, and satisfies the boundary conditions
\begin{equation}
h(x)=\begin{cases}
a & \si x\el\AA,\\
b & \si x\el\BB.
\end{cases}\label{eq:BoundaryCondAB}
\end{equation}
We will let (see Figure \ref{fig:Cross-section1}) $\UU=\HH$, where
\begin{equation}
\HH(x)=\nabla h(x)/||\nabla h(x)||^{2}.\label{eq:HarmonicField}
\end{equation}
This field generates the channel $\CH$ and has $h$ as projection
function. We will refer to $h$ and $\HH$ as a \emph{natural projection
function} and a \emph{natural generating field} for the channel $\CH$
with lateral cross sections $\AA$ and $\BB$. 
\begin{rem}
If we write $h=h_{a,b}$ to specify that $h$ takes the values $a$
and $b$ in $\AA$ and $\BB$, respectively, then we have that
\[
h_{a,b}=a+(b-a)h_{0,1}.
\]
\end{rem}

\subsection{Flux functions}

\begin{figure}
\includegraphics[scale=0.25]{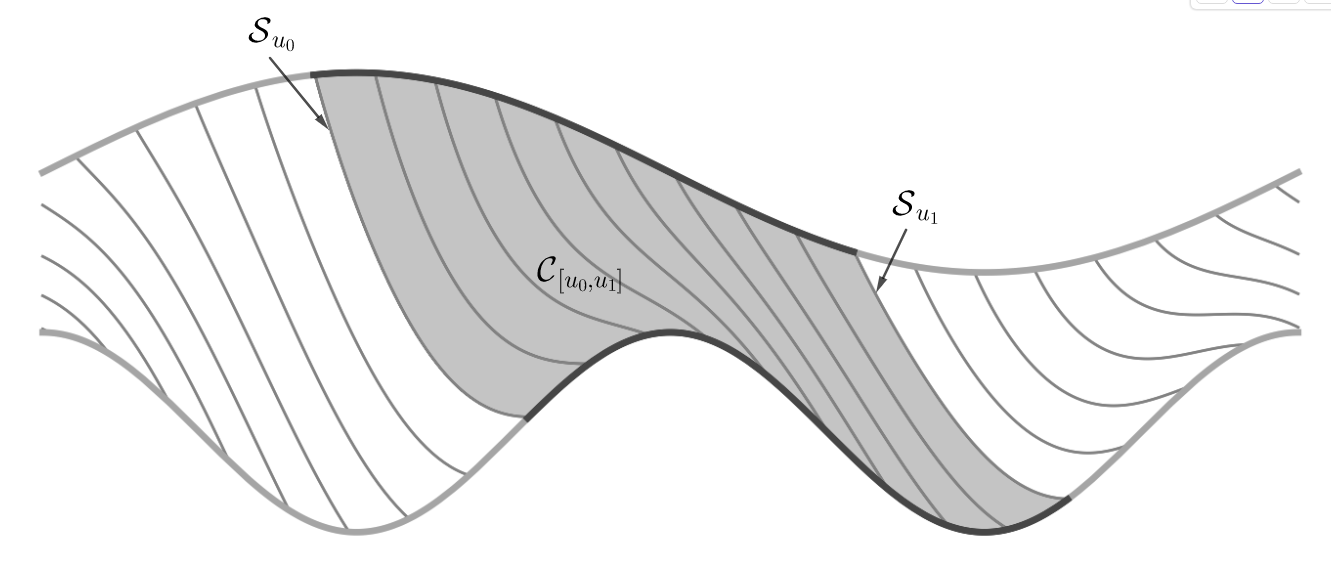}\caption{\label{fig:Channel-region-between}Channel region between two cross
section}
\end{figure}
The \emph{flux function} of a vector field $\VV$ in channel $\CH$
is defined as \footnote{Mathematically, this is the integral over $\SS_{u}$ of the component
of $\VV$ normal to $\SE_{u}$ }
\[
\fl{\VV}(u)=\text{flux of }\VV\text{ across }\SS_{u}.
\]
In the above definition we have assumed that we have fixed a generating
vector field $\UU$ for $\CH$ (and hence its cross sections). The
importance of the generating vector field $\UU$ of a channel is that
we will be able to express many quantities of interest in terms $\UU$.
In particular we will consider the flux functions of vector fields
$\VV$ of the form $\VV=\lambda\UU$, where $\lambda$ is a scalar
valued function in $\CH$. In this case we have that
\[
\fl{\VV}(u)=\der{c_{\lambda}}u(u),
\]
where for (see Figure \ref{fig:Channel-region-between} )
\[
\CH_{[u_{0},u_{1}]}=\text{union of the sets }\SE_{u}\text{ for }u_{0}\leq u\leq u_{1}
\]
we defined\footnote{The total concentration is obtained by integrating the function $\lambda$
in the region $\CH_{[0,u]}$.} 
\[
c_{\lambda}(u)=\text{total concentration of }\lambda\text{ in }\CH_{[0,u]}.
\]

\begin{example}
Let 
\[
\nu(u)=\text{volume of }\CH_{[0,u]}.
\]
For $\lambda=1$ we have that $\nu(u)=c_{\lambda}(u)$, and hence
\[
\der{\nu}u(u)=\fl{\UU}(u).
\]
\end{example}

An important property of the flux function, that we will use frequently,
is that if we can write\footnote{If we think of $\lambda$ as a function of $x$ (i.e $\lambda=\lambda(x))$
the expression $\lambda=\lambda(u)$ means that $\lambda(x)=\rho(u(x))$
for a scalar valued function $\rho=\rho(u)$. Our notation has the
advantage of avoiding the use the extra function $\rho$.} $\lambda=\lambda(u)$, then 
\[
\fl{\lambda V}(u)=\lambda(u)\fl V(u)
\]

\section{Effective diffusion on channels}

For a given channel $\CH$, we are interested in studying the evolution
of a density function $P=P(x,t)$ that obeys the diffusion equation
\[
\pd Pt(x,t)=D_{0}\Delta P(x,t).
\]
We will assume reflective boundary conditions on the wall $\WA$ of
$\CH$, i.e the gradient of $P$ has no flux across $\WA$. Using
a projection function in $\CH$, we will try to reduce the above equation
to a diffusion equation in a $1$-dimensional spatial variable. To
do this, we define the \emph{total concentration function} as
\[
c(u,t)=\text{total concentration of density }P\text{ in }\CH_{[0,u]}\text{ at time }t.
\]
and the \emph{volume}\footnote{When we speak of volume in $n$-dimensional space we are referring
to $n$-dimensional volume, i.e length for $n=1$, area for $n=2$,
etc.} function $\nu$ by
\begin{equation}
\nu(u)=\text{volume of }\CH_{[0,u]}.\label{eq:volumefunction}
\end{equation}
We can now define the \emph{effective concentration} as
\[
p(u,t)=\pd cu(u,t)/\der{\nu}u(u).
\]
If we let $u=u(\nu)$ be the value of $u$ that corresponds to volume
$\nu$ and $p(\nu,t)=p(u(\nu),t)$, then
\[
p(\nu,t)=\pd c{\nu}(\nu,t).
\]
Hence, the total concentration of $P$ in the region $\CH_{[u_{1},u_{2}]}$
at time $t$ is given by
\[
\int_{\nu_{1}}^{\nu_{2}}p(\nu,t)d\nu\for\nu_{i}=\nu(u_{i}).
\]

\begin{rem}
In most of the literature the effective concentration is defined as
\begin{equation}
p(u,t)=\der cu(u,t),\label{eq:ClassicEffectiveDensity}
\end{equation}
so that the total concentration of $P$ in $\CH_{[u_{1},u_{2}]}$
is 
\[
\int_{u_{1}}^{u_{2}}p(u,t)du.
\]
From a mathematical point of view the definition of the effective
concentration \ref{eq:ClassicEffectiveDensity} is not convenient
for the following reason. If we introduce a new variable $v=v(u)$,
with our definition of effective concentration we have that 
\begin{align*}
p(u,t) & =\pd cu(u,t)/\left(\der{\nu}u(u)\right)\\
 & =\left(\pd cv(v(u),t)\der vu(u)\right)/\left(\der{\nu}v(v(u))\der vu(u)\right)\\
 & =\pd cv(v(u),t)\der{\nu}v(v(u))\\
 & =p(v(u),t).
\end{align*}
This is the proper formula for the change of variable of a function.
On the other hand, if we define $p$ as in \ref{eq:ClassicEffectiveDensity}
we have
\[
p(u,t)=\der cu(u,t)=\der{\nu}u(u)\der c{\nu}(\nu(u),t)=\der{\nu}u(u)p(\nu(u),t),
\]
which is the way vectors (not functions) transform under a change
of variable. 
\end{rem}

\subsection{Infinite transversal diffusion rate}

If we assume that the density function $P$ stabilizes infinitely
fast along the cross sections $\SS_{u}$ of $\CH$, which is equivalent
to $P$ being constant along them, we arrive at the \emph{effective
diffusion equation} (see Section \ref{subsec:Infinite-transversal-diffusion})

\begin{equation}
\pd pt(u,t)=\dv(\dc(u)\nabla p(u,t)),\label{eq:EffecitveDiffEq}
\end{equation}
where the \emph{effective diffusion coefficient} is given by\footnote{In the expression for flux function in this formula $u$ is being
interpreted as a projection function on the channel $\CH$, and its
gradient computed with respect to the metric in $M$.}
\begin{equation}
\dc(u)=D_{0}\fl{\nabla u}(u)\der{\nu}u(u).\label{eq:EffectiveDiffusionIDR}
\end{equation}
The divergence and gradient operators $\dv$ and $\nabla$ in formula
\ref{eq:EffecitveDiffEq} are the ones associated to the metric\footnote{This metric defines a distance function $d$ between values $u_{1}$
and $u_{2}$, given by
\[
d(u_{1},u_{2})=\int_{u_{1}}^{u_{2}}\sqrt{g(u)}du=\text{\ensuremath{\nu(u_{2})-\nu(u_{1})},}
\]
which is the volume of the region $\CH_{[u_{1},u_{2}]}$. }
\begin{equation}
g(u)=\left(\der{\nu}u(u)\right)^{2},\label{eq:AreaMetric}
\end{equation}
i.e
\[
\nabla p=\frac{1}{g}\pd pu\and\nabla\cdot j=\frac{1}{\sqrt{g}}\pdo u\left(\sqrt{g}j\right).
\]
Observe that if we let $u=\nu$ then $g=1,$ and hence
\[
\nabla p=\pd p{\nu},\dv j=\pd j{\nu}
\]
and
\[
\dc(\nu)=D_{0}\fl{\nabla\nu}.
\]

\begin{rem}
The effective diffusion coefficient $\dc$ connects the \emph{effective
flux}
\[
j(u,t)=\fl{J_{t}}(u)/\der{\nu}u(u)\where J_{t}=D_{0}\nabla P_{t}
\]
with the gradient of the effective density function. More concretely,\emph{
Fick's first law} establishes that (see sections \ref{subsec:The-effective-continuity}
and \ref{subsec:Infinite-transversal-diffusion})
\[
j(u,t)=-\dc(u)\nabla p(u,t).
\]
\begin{rem}
The condition of $P$ being constant along the cross sections of the
channel implies that we can write (see Section \ref{subsec:Infinite-transversal-diffusion})
\[
P(x,t)=p(u(x),t),
\]
where $p$ is the effective density function and $u$ is the projection
function. If we had used the definition of effective concentration
found in most of the literature, this identity would not hold.
\end{rem}

\end{rem}

\subsubsection*{Comparison with the generalized Fick-Jacobs equation}

If we let
\[
\sigma(u)=\der{\nu}u(u)\and p_{f}(u,t)=\der cu(u,t),
\]
we can write equation \ref{eq:EffecitveDiffEq} as a \emph{generalized
Fick-Jacobs equation}
\[
\der{p_{f}}t(u,t)=\pdo u\left(\sigma(u)\dc_{f}(u)\pdo u\left(\frac{p_{f}(u,t)}{\sigma(u)}\right)\right),
\]
where the effective diffusion coefficient is given by
\[
\dc_{f}(u)=\frac{D_{0}\fl{\nabla u}(u)}{\sigma(u)}.
\]
If we define the effective flux $j_{f}$ as
\[
j_{f}(u,t)=\fl{J_{t}}(u,t)
\]
then we have the continuity equation (see \ref{subsec:The-effective-continuity})
\[
\pd{p_{f}}t+\pd{j_{f}}u=0.
\]
Using the Fick-Jacobs equation we conclude that
\[
j_{f}=-\sigma\dc_{f}\pdo u\left(\frac{p_{f}}{\sigma}\right).
\]
From these equations and the formulas
\[
p=p_{f}/\sigma\and j=j_{f}/\sigma
\]
we obtain
\[
j(u,t)=-\dc_{f}(u)\pd pu(u,t).
\]
We conclude that the difference between the effective diffusion coefficient
given by the generalized Fick-Jacobs equation and ours is that: in
the first case the gradient used in Fick's first law is that associated
to the metric $g=1$, and in the second case it is that associated
to the metric $g(u)=\sigma(u)^{2}$. The formula connecting both coefficients
is 
\[
\dc_{f}(u)=\frac{\dc(u)}{\sigma(u)^{2}}.
\]
Observe that when the cross sections of $\CH$ are parametrized by
the volume variable $\nu$, we have that $\sigma(\nu)=1$ and hence
\[
\dc_{f}(\nu)=\dc(\nu)=D_{0}\fl{\nabla\nu}(\nu).
\]
Furthermore, in this case both the effective diffusion equation and
the Fick-Jacobs equation become the diffusion equation 
\[
\pd pt(\nu,t)=\pdo{\nu}\left(\dc(\nu)\pd p{\nu}(\nu,t)\right).
\]

\subsubsection*{Cross section density function}

If we define the area\footnote{We refer to area as $(n-1)$-dimensional area. For $n=2$ this means
length, for $n=3$ this means area in the usual sense, etc. By convention
we speak of volume when we want to measure the ``extent'' of $n$-dimensional
objects in $n$-dimensional space, and area when we want to measure
the ``extent'' of $(n-1)$-dimensional objects in $n$-dimensional
space.} function as
\[
\AA(u)=\text{area of }\SS_{u}
\]
and let
\[
\GG(u)=\frac{\fl{\nabla u}(u)}{\AA(u)},
\]
then we can write
\[
\dc(u)=D_{0}\AA(u)\GG(u)\der{\nu}u(u).
\]
Since the cross sections $\SE_{u}$ are the level sets of $u$, the
vector field $\nabla u$ is orthogonal to them. Hence, if we orient
the cross sections $\SS_{u}$ so that their normal fields have the
same direction as $\nabla u$, we have that
\[
\GG(u)=\text{average value of }|\nabla u|\text{ on }\SE_{u}.
\]
This number measures the average density of cross sections near $\SS_{u}$,
and we will refer to $\GG$ as the \emph{cross section density function}.
If the cross sections of the channel are parametrized by the volume
variable $\nu$, we have that
\[
\dc(\nu)=D_{0}\AA(\nu)\GG(\nu).
\]

\subsection{Finite transversal diffusion rate}

Consider a channel $\CH$ whose cross sections $\SE_{u}$ are generated
by a vector field $\UU$, and let us drop the assumption that the
density function $P=P(x,t)$ stabilizes infinitely fast along these
cross sections. To give a formula for the effective diffusion coefficient
$\dc=\dc(u)$, we will make use of the natural projection function
$h$ and the natural generating field $\HH$ of the channel $\CH$
with lateral cross sections $\SS_{u_{0}}$ and $\SS_{u_{1}}$ for
$u_{0}<u_{1}$(see section \ref{subsec:Harmonic-projection-maps}).
In this context, we will refer to the projection function $u$ and
the field $\UU$ as the \emph{imposed projection function and field}
(to distinguish them from the natural ones: $h$ and $\HH$).

Let $\rho$ be the effective density function of $h$ under the projection
map $u$, i.e
\[
\rho(u)=\fl{hU}(u)/\der{\nu}u(u).
\]
In section  \ref{subsec:Finite-transversal-diffusion} we proved that
(for $u$ with $u_{0}\leq u\leq u_{1}$) the effective diffusion coefficient
$\dc$ appearing in formula \ref{eq:EffecitveDiffEq} can be computed
as
\begin{equation}
\dc(u)=\JJ\left(\der{\nu}u(u)\right)^{2}/\fl{\lambda U}(u),\label{eq:DiffCoeffFiniteTDRC}
\end{equation}
where $\lambda=\lambda(x)$ is a scalar valued function in $\CH$
defined by\footnote{The gradient and divergence operator appearing in this formula are
computed with respect to the metric in $M$.}
\begin{equation}
\lambda=\nabla h\cdot U+(h-\rho\circ u)\nabla\cdot U\label{eq:LambdaDiffussion}
\end{equation}
and the constant $\JJ$ is given by\footnote{By using the fact that $h$ is harmonic we showed in \ref{subsec:Finite-transversal-diffusion}
that the function $\JJ(u)=\fl{\nabla h}(u)$ is a constant function,
i.e independent of $u.$ }
\[
\JJ=D_{0}\fl{\nabla h}(u_{0}).
\]

\subsubsection*{Channels with natural projection map}

If for a given channel $\CH$ we choose the imposed projection map
an generating vector field to be the natural ones, i.e 
\[
\UU=\HH\and u=h,
\]
then we have that
\[
\nabla h\cdot U=\frac{\nabla h\cdot\nabla h}{||\nabla h||^{2}}=1.
\]
and
\begin{align*}
\rho(u(x)) & =\fl{hU}(u(x))/\der{\nu}u(u(x))\\
 & =h(x)\fl U(u(x))/\der{\nu}u(u(x))\\
 & =h(x),
\end{align*}
where we have made use of the formula
\[
\der{\nu}u(u)=\fl U(u).
\]
Hence $\lambda=1$ in formula \ref{eq:LambdaDiffussion}, which implies
\[
\dc(u)=\JJ\left(\der{\nu}u(u)\right)^{2}\left(\fl U(u)\right)^{-1}=D_{0}\fl{\nabla u}(u)\der{\nu}u(u).
\]
We conclude that\emph{ when using the natural projection function
and field of a channel, the formulas for the effective diffusion coefficient
in the finite and infinite diffusion transversal rate cases coincide. }

\section{Derivation of the effective diffusion coefficient formula}

Let $M$ be an oriented Riemannian manifold of dimension $n$. We
are interested in the diffusion equation\footnote{Given the variety of mathematical objects that we will use, throughout
this section we won't follow the convention of using bold face to
denote non-scalar quantities.}
\[
\pd Pt(x,t)=\dv(D(x)\nabla P(x,t)),
\]
where $P:M\times\RR\rightarrow\RR$ is a time dependent function in
$M$ and $D(x):T_{x}X\rightarrow T_{x}X$ is a linear map for every
$x$ in $M$. The divergence and gradient operators in the above formula
can defined in terms of exterior algebra operations as
\[
\dv J=*d*J^{\flat}\and\nabla P=(dP)^{\#},
\]
where $d:\bigwedge^{k}M\rightarrow\bigwedge^{k+1}M$ is the exterior
derivative, $*:\bigwedge^{k}M\rightarrow\bigwedge^{n-k}M$ the Hodge
star operator, and the musical isomorphisms $\sharp$ and $\flat$
allow us to identify 1-forms and vector fields. If we let $g$ stand
for the metric tensor in $M$ and use local coordinates $x_{1},\ldots x_{n}$,
we can write
\[
\dv J=\frac{1}{|g|^{1/2}}\sum_{i=1}^{n}\pdo{x_{i}}\left(|g|^{1/2}J_{i}\right)\where|g|=\det(g)
\]
and
\[
(\nabla P)^{i}=\sum_{j=1}^{n}g^{ij}\pd P{x_{j}}\where(g^{ij})=(g_{ij})^{-1}.
\]
\[
.
\]
In a homogeneous and isotropic medium the diffusion has the form
\begin{equation}
\pd Pt(x,t)=D_{0}\Delta P(x,t)\where\Delta=*d*d\and D_{0}\in\RR.\label{eq:FullDiffusionEquation}
\end{equation}

\subsection{Channels and projection functions}

Let $M$ be an $n$-dimensional oriented Riemannian manifold. We will
say that $\CH\subset M$ is generated by a vector field $U$, if $\CH$
is the union of phase curves of $U$ that have transversal intersection
with an $(n-1)$-dimensional sub-manifold with boundary $\SS_{0}$.
We will then say that $\CH$ is a channel generated by $U$ having
$\SE_{0}$ as an \emph{initial cross section}. A smooth function $u:\CH\rightarrow\RR$
is a projection function for the field $U$ if $du(U)=1$. We will
usually choose $\SE_{0}$ so that $\SE_{0}=u^{-1}(0)$. The cross
section $\SE_{s}$ of $\CH$ at $s$ is defined by the formula
\[
\SE_{s}=u^{-1}(s).
\]
Recall that the phase flow $\{\varphi_{s}:\CH\rightarrow\CH\}_{s\in\RR}$
of $U$ is defined by
\[
\left.\do s\right|_{s=0}(\varphi_{s}(x))=U(x),
\]
and satisfies
\begin{equation}
\varphi_{s_{1}+s_{2}}=\varphi_{s_{1}}\circ\varphi_{s_{2}}.\label{eq:GroupProperty}
\end{equation}
The condition $du(U)=1$ is then equivalent to
\[
u(\varphi_{s}(x))=u(x)+s,
\]
and hence
\[
\SE_{s+h}=\varphi_{h}(\SE_{s}).
\]
If we let $\WA=\partial\CH$ then $\WA$ is the union of phase curves
of $U$ that intersect $\partial\SE_{0}$. We will refer to $\WA$
as the reflective wall of $\CH$. We define
\[
\CH_{[s_{1},s_{2}]}=u^{-1}([s_{1},s_{2}])\and\WA_{[s_{1},s_{2}]}=\WA\cap\CH_{[s_{1},s_{2}]}.
\]

\subsubsection*{Flux functions}

We will let $\mu\in\bigwedge^{n}M$ stand for global volume form associated
with the metric in $M$. The orientation in $\CH$ will the the one
induced by the orientation of $M$, i.e we will let the orientation
form be the one obtained by restricting $\mu$ to $\CH$. Observe
that
\begin{equation}
du\wedge\iota_{U}(\mu)=du\wedge(*U^{\flat})=<du,U^{\flat}>\mu=du(U)\mu=\mu.\label{eq:MuSplit}
\end{equation}
If we let $i_{u}:\SE_{u}\rightarrow\CH$ be the inclusion map then
$i_{u}^{*}(\iota_{U}(\mu))$ is an $(n-1)$-form in $\SE_{u}$ which
vanishes no-where in $\SE_{u}$. We will use this form as an orientation
form for $\SS_{u}$. For a vector field $V$ in $\CH$ we define the
flux function $\fl V:\RR\rightarrow\RR$ as 
\[
\fl V(u)=\int_{\SS_{u}}\iota_{V}(\mu)=\int_{\SE_{u}}*(V^{\flat}).
\]
In particular
\[
\fl{\nabla P}(u)=\int_{\SE_{u}}*(dP).
\]

\subsubsection*{Change of variable formulas}

Let $u:\CH\rightarrow\RR$ be a projection function for $U$ and $f:\RR\rightarrow\RR$
a function with positive derivative. If we let $v=f\circ u$ then
\[
dv(x)=f'(u(x))du(x),
\]
which implies that $v$ is a projection function for the field $V$
defined by $V(x)=U(x)/f'(u(x))$. To simplify notation, we will write
the conditions $v=f\circ u$ and $u=f^{-1}\circ v$ as 
\[
v=v(u)\and u=u(v),
\]
where in the first equation $u$ is seen as a scalar value and $v$
as a function, and on the second formula $v$ is seen as a scalar
value and $u$ as a function. If we denote a cross sections at $u$
as $\SE_{u}$ and a cross section at $v$ as $\SS_{v}$, then the
formulas $u^{-1}(s)=v^{-1}(f(s))$ and $v^{-1}(s)=u^{-1}(f^{-1}(s))$
can be simply written as $\SE_{u}=\SS_{v(u)}$ and $\SS_{v}=\SS_{u(v)}$.
Furthermore, we have that
\begin{equation}
dv=\left(\der vu\right)du\and V=\left(\der vu\right)^{-1}U,\label{eq:ChangeOfCoords}
\end{equation}
where 
\[
\der vu(x)=f'(u(x)).
\]

\begin{rem}
If for a positive function $\lambda:\RR\rightarrow\RR$ we let $V=(\lambda\circ u)U$,
then we can recover the projection function $v=v(u)$ for $V$ as
\[
v(u)=v_{0}+\int_{u_{0}}^{u}\left(\frac{1}{\lambda(s)}\right)ds\for v_{0}\in\RR.
\]
\end{rem}

\subsection{Some useful identities}

We will now derive some identities that will be useful in our study
of diffusion processes on channels. Let $\CH$ be a channel generated
by a field $U$ and with a projection function $u$. In what follows
we will make use of Cartan's magic formula
\[
\LL_{U}=\iota_{U}\circ d+d\circ\iota_{U},
\]
where $\LL_{U}$ is the Lie derivative with respect to $U$. 
\begin{lem}
\label{lem:Integrals}If $\alpha$ is an $(n-1)$-form in $\CH$ and
we define
\[
f(u)=\int_{S_{u}}\alpha,
\]
then 
\[
f'(u)=\int_{S_{u}}\LL_{U}\alpha.
\]
If $\omega$ is an $n$-form in $\CH$ and for any $u_{0}\in\text{\ensuremath{\RR}}$
we define
\[
g(u)=\int_{\CH_{[u_{0},u]}}\omega
\]
then
\[
g'(u)=\int_{S_{u}}\iota_{U}(\omega)
\]
and 
\[
g''(u)=\int_{S_{u}}\left(d\lambda(U)+\lambda\dv U\right)\iota_{U}(\mu)\where\lambda=*\omega.
\]
\end{lem}

\begin{proof}
From the formula $S_{u+h}=\varphi_{h}(S_{u})$ we obtain
\[
f(u+h)-f(u)=\int_{S_{u+h}}\alpha-\int_{S_{u}}\alpha=\int_{S_{u}}(\varphi_{h}^{*}\alpha-\alpha),
\]
and hence
\[
f'(u)=\lim_{h\mapsto0}\int_{S_{u}}\frac{1}{h}(\varphi_{h}^{*}(\alpha)-\alpha)=\int_{S_{u}}\LL_{U}\alpha.
\]
To prove the second part of the lemma observe that
\[
g(u+h)-g(u)=\int_{\CH_{[u,u+h]}}\omega,
\]
and hence
\[
g'(u)=\lim_{h\mapsto0}\frac{g(u+h)-g(u)}{h}=\lim_{h\mapsto0}\frac{1}{h}\int_{u}^{u+h}\int_{\SE_{t}}(\iota_{U}(\omega))dt=\int_{\SE_{u}}\iota_{U}(\omega).
\]
Combining the previous results we obtain 
\[
g''(u)=\int_{S_{u}}\LL_{U}(\iota_{U}(\omega)).
\]
Using Cartan's magic formula it is easy to verify that
\[
\LL_{U}(\iota_{U}(\omega))=\iota_{U}(\LL_{U}(\omega)).
\]
We can write $\omega=\lambda\mu$ for $\lambda=*\omega$, and hence
\[
\LL_{U}(\omega)=\LL_{U}(\lambda\mu)=\iota_{U}(d\lambda)\mu+\lambda\LL_{U}\mu.
\]
Using this and the fact that $\LL_{U}\mu=(\dv U)\mu$, we conclude
that
\[
\iota_{U}(\omega)=(d\lambda(U)+\lambda(\dv U))\iota_{U}(\mu)
\]
\end{proof}

\subsection{\label{subsec:The-effective-continuity} The effective continuity
equation}

If we let the metric tensor in the $u$ variable be 
\[
g(u)=\left(\der{\nu}u(u)\right)^{2},
\]
then the divergence and gradient operators are given by the formulas
\[
\nabla\cdot j=g^{-1/2}\pdo u\left(g^{1/2}j\right)\and\nabla p=g^{-1}\pd pu.
\]
Consider a concentration function $P=P(x,t)$ and the flux vector
field $J=J(x,t)$. Let us write $P_{t}(x)=P(x,t)$ and $J_{t}(x)=J(x,t)$,
and for a channel $\CH$ define the \emph{effective flux function}
as
\[
j(u,t)=\fl{J_{t}}(u)/\der{\nu}u(u)\where\fl{J_{t}}=\int_{\SE_{u}}*J_{t}^{\flat}.
\]
and the \emph{effective concentration} as
\[
p(u,t)=\pd cu(u,t)/\der{\nu}u(u)\where c(u,t)=\int_{\CH_{[0,u]}}*P_{t}.
\]
By Lemma \ref{lem:Integrals} we have that
\[
\pd cu(u,t)=\int_{\SE_{u}}\iota_{U}(*P_{t})
\]
and
\[
\der{\fl{J_{t}}}u(u)=\int_{\SS_{u}}\LL_{U}(*J_{t}^{\flat})=\int_{\SS_{u}}(d\circ\iota_{U}+\iota_{U}\circ d)(*J_{t}^{\flat}).
\]
If we assume reflective boundary conditions on the wall $\WA$ of
$\CH$, we get
\[
\int_{S_{u}}d(\iota_{U}(*J_{t}^{\flat}))=\int_{\partial S_{u}}\iota_{U}(*J_{t}^{\flat})=0.
\]
Using the above formulas and the continuity equation 
\[
*\pd Pt(x,t)+d*J^{\flat}(x,t)=0
\]
we obtain
\[
\der{\fl{J_{t}}}u(u)=\int_{S_{u}}(\iota_{U}\circ d)(*J_{t}^{\flat})=-\int_{S_{u}}\iota_{U}\left(*\pd Pt\right),
\]
and hence 
\[
\int_{S_{u}}\iota_{U}\left(*\pd Pt\right)=\pdo t\int_{S_{u}}\iota_{U}(*P_{t})=\pdo t\left(\pd cu(u,t)\right).
\]
We conclude that
\[
\pdo t\left(\pd cu(u,t)\right)+\der{\fl{J_{t}}}u(u)=0,
\]
which implies that

\[
\pdo t\left(g^{-1/2}(u)\pd cu(u,t)\right)+g^{-1/2}(u)\pdo u\left(g^{1/2}(u)g^{-1/2}(u)\der{\fl{J_{t}}}u(u)\right)=0.
\]
This last equation is known \emph{effective continuity equation} and
can be re-written as
\begin{equation}
\pd pt(u,t)+\dv j(u,t)=0.\label{eq:EffectiveContinuity}
\end{equation}

\subsection{\label{subsec:Infinite-transversal-diffusion}Infinite transversal
diffusion rate}

The assumption of an infinite transversal diffusion rate is expressed
mathematically by letting 
\[
P(x,t)=\rho(u(x),t)
\]
for a function $\rho:\RR\times\RR\rightarrow\RR$. The effective density
function can then be written as
\[
p(u,t)=g(u)^{-1/2}\int_{S_{u}}u^{*}(\rho_{t})\iota_{U}\mu=\rho(u,t)g(u)^{-1/2}\int_{S_{u}}\iota_{U}(\mu).
\]
From the formula 
\[
\int_{S_{u}}\iota_{U}(\mu)=\der{\nu}u(u)=g(u)^{1/2},
\]
we conclude that
\[
p(u,t)=\rho(u,t).
\]
Using Fick's law $J_{t}=-D_{0}\nabla P_{t}$, we obtain
\[
\fl{J_{t}}(u)=\int_{S_{u}}*J_{t}^{\flat}=-D_{0}\int_{S_{u}}*(dP_{t})=-D_{0}\int_{S_{u}}*(u^{*}(d\rho_{t})).
\]
Since (for $s$ equal to the identity map in $\RR$)
\[
u^{*}(d\rho_{t})=u^{*}\left(\pd{\rho_{t}}sds\right)=u^{*}\left(\pd{\rho_{t}}s\right)du,
\]
we obtain
\[
\fl{J_{t}}(u)=-D_{0}\pd{\rho}u(u,t)\int_{S_{u}}*(du).
\]
Using this last formula and the fact that $\rho=p$, we obtain
\[
j(u,t)=g(u)^{-1/2}\fl{J_{t}}(u)=-\left(D_{0}g(u)^{1/2}\int_{S_{u}}*(du)\right)g(u)^{-1}\pd pu(u,t).
\]
Substitution of this formula for $j$ in the effective continuity
equation \ref{eq:EffectiveContinuity} leads to the \emph{effective
diffusion formula}
\begin{equation}
\pd pt(u,t)=\dv(\dc(u)\nabla p(u,t)),\label{eq:EffectiveDiffusionICoordFree}
\end{equation}
where the \emph{effective diffusion coefficient} is given by 
\begin{align*}
\dc(u) & =D_{0}\left(\int_{S_{u}}*(du)\right)g(u)^{1/2}.\\
 & =D_{0}\fl{\nabla u}(u)\der{\nu}u(u).
\end{align*}

\subsection{\label{subsec:Finite-transversal-diffusion}Finite transversal diffusion
rate}

We will now consider the case when density function $P=P(x,t)$ is
not necessarily constant along the cross sections of the channel.
In general it is not possible to define define $\dc=\dc(u)$ such
that the effective density function $p=p(u,t)$ satisfies the $1$-dimensional
diffusion equation \ref{eq:EffectiveDiffusionICoordFree} exactly,
but for many cases of narrow channels it is possible to find $\dc$
such that $p$ satisfy \ref{eq:EffectiveDiffusionICoordFree} to a
very good approximation. In any case, if such a $\dc$ existed we
could recover it from of a stable solution $\rho=\rho(u)$ to \ref{eq:EffectiveDiffusionICoordFree}.
In fact, if $\rho$ is such a function we have that
\[
\dv(\dc\nabla\rho)=0,
\]
which is equivalent to
\begin{equation}
\pdo u\left(\frac{\dc}{\sigma}\der{\rho}u\right)=0\where\sigma=g^{1/2}=\der{\nu}u.\label{eq:StableSol}
\end{equation}
Hence, we can find a constant $\JJ\in\RR$ such that
\begin{equation}
\dc(u)=\JJ\sigma(u)\left(\der{\rho}u(u)\right)^{-1}.\label{eq:DInfiniteTR}
\end{equation}

\begin{rem}
If we introduce a new variable $v=v(u)$, then we have that
\[
\dc(v)=\dc(u(v)),
\]
since
\begin{align*}
\dc(v) & =\JJ\der{\nu}v(v)\left(\der{\rho}v(v)\right)^{-1}\\
 & =\JJ\der{\nu}u(u(v))\der uv(v)\left(\der{\rho}u(u(v))\der uv(v)\right)^{-1}\\
 & =\dc(u(v)).
\end{align*}
\end{rem}

We will now assume that $\rho$ is the effective concentration function
of a stable solution $h=h(x)$ to the full diffusion equation \ref{eq:FullDiffusionEquation}
(with reflective boundary conditions on $\WA$). We then have that
\[
\rho(u)=\frac{1}{\sigma(u)}\int_{\SE_{u}}h\iota_{U}(\mu)=\frac{1}{\sigma(u)}\der cu(u)
\]
for 
\[
c(u)=\int_{\CH_{[0,u]}}h\mu,
\]
and hence
\[
\der{\rho}u=\frac{d}{du}\left(\frac{1}{\sigma}\der cu\right)=\frac{1}{\sigma}\left(\dder cu-\rho\dder{\nu}u\right).
\]
Using Lemma \ref{lem:Integrals} we obtain
\[
\dder cu=\int_{\SE_{u}}(dh(U)+h\dv U)\iota_{U}(\mu),
\]
and since 
\[
\nu(u)=\int_{\CH_{[0,u]}}\mu
\]
then
\[
\dder{\nu}u=\int_{\SE_{u}}(\dv U)\iota_{U}(\mu).
\]
We conclude that
\[
\dc(u)=\JJ\sigma^{2}(u)/\fl{\lambda\UU}(u)=\JJ\fl{\UU}(u)\left(\frac{\fl U(u)}{\fl{\lambda U}(u)}\right)
\]
where
\[
\lambda=dh(U)+(h-u^{*}(\rho))\dv U.
\]

\subsubsection*{Computation of $\protect\JJ$}

By definition, we have\footnote{Apparently $\JJ$ depends on $u$, but we will show below that $\JJ$
is actually a constant function (as required for the formula we computed
for the effective diffusion coefficient $\dc$) }
\[
\JJ(u)=\frac{\dc(u)}{\sigma(u)}\der{\rho}u(u).
\]
Using Fick's laws
\begin{align*}
j(u) & =-\dc(u)\nabla\rho(u)\\
J(x) & =-D_{0}\nabla h(x)
\end{align*}
and the formulas
\begin{align*}
j(u) & =\frac{1}{\sigma(u)}\int_{\SE_{u}}*J^{\flat}\\
\nabla\rho(u) & =\frac{1}{\sigma(u)^{2}}\der{\rho}u(u)
\end{align*}
we obtain
\[
\JJ(u)=D_{0}\int_{\SE_{u}}*(dh)=D_{0}\fl{\nabla h}(u).
\]
The function $\JJ=\JJ(u)$ is in fact a constant function (i.e independent
of $u$), since for any two values $u_{1}$ and $u_{2}$ we have that
(by Stokes Theorem and the reflective boundary conditions on $\WA)$
\[
\JJ(u_{2})-\JJ(u_{1})=D_{0}\int_{S_{u_{2}}-S_{u_{1}}}(*dh)=\int_{\CH_{[u_{1},u_{2}]}}d*dh=\int_{\text{\ensuremath{\CH}}_{[u_{1},u_{2}]}}*\Delta h=0.
\]

\subsubsection*{Lateral boundary conditions}

It is \emph{important} to notice that formula \ref{eq:DInfiniteTR}
holds only under the assumption that $\rho'(u)\not=0$ for all $u\in\RR$.
We can achieve this if for $\alpha\not=\beta$ we fix boundary the
conditions 
\begin{equation}
\rho(a)=\alpha\and\rho(b)=\beta.\label{eq:LateralBoundaryConditions}
\end{equation}
For fixed values of $\alpha$ and $\beta$ we will denote the stable
solution to \ref{eq:EffectiveDiffusionICoordFree} satisfying these
boundary conditions by $\rho_{\alpha,\beta}$. Using the linearity
of equation \ref{eq:StableSol} we obtain
\[
\rho_{\alpha,\beta}=\alpha+(\beta-\alpha)\rho_{0,1}.
\]
If we denote the constant $\JJ$ associated to $\rho_{\alpha,\beta}$
by $\JJ(\alpha,\beta)$ then
\[
\dc=\sigma\JJ(\alpha,\beta)\left(\der{\rho_{\alpha,\beta}}u\right)^{-1}.
\]
Since $\dc$ is independent of the choice of $\alpha$ and $\beta$
we must have
\[
\JJ(\alpha,\beta)\left(\der{\rho_{\alpha,\beta}}u\right)^{-1}=\JJ(0,1)\left(\der{\rho_{0,1}}u\right)^{-1},
\]
from which we obtain the formula
\[
\JJ(\alpha,\beta)=(\beta-\alpha)\JJ(0,1).
\]
The boundary conditions \ref{eq:LateralBoundaryConditions} can be
written in terms of $H$ (using Lemma \ref{lem:Integrals}) as
\begin{align*}
\frac{1}{\sigma(a)}\int_{\SE_{a}}h\iota_{U}(\mu) & =\alpha,\\
\frac{1}{\sigma(b)}\int_{\SE_{b}}h\iota_{U}(\mu) & =\beta.
\end{align*}
If we choose $h$ so that it is has constant value $h_{a}$ in $\SS_{a}$
and constant value $h_{b}$ in $\SS_{b}$, the above conditions become
\[
h_{a}=\alpha\and h_{b}=\beta.
\]

\subsection{Channels defined by harmonic conjugate functions}

Let $M$ be a 2-dimensional oriented surface. We will say that $u,v:M\rightarrow\RR$
are harmonic conjugate if
\[
dv=*du,
\]
or equivalently
\[
\nabla v=i\nabla u.
\]
Observe that in this case
\[
*dv=**du=-du.
\]
The existence of a harmonic conjugate $v$ for $u$ implies that $u$
and $v$ are harmonic, since
\begin{align*}
\Delta u & =*d*du=*(d^{2}v)=0,\\
\Delta v & =*d*dv=-*(d^{2}u)=0.
\end{align*}
For fixed value $v_{1},v_{2}\in\RR$, consider a channel $\CH$ defined
as
\[
\CH=\{x\in M|v_{1}\leq v(x)\leq v_{2}\},
\]
If we use a harmonic conjugate $u$ of $v$ as projection function
for this channel, then $u$ is a harmonic function with reflective
boundary conditions on $\WA$. The channels $\CH$ has generating
field
\[
U=\frac{\nabla u}{|\nabla u|^{2}}.
\]
The effective diffusion coefficient both in the infinite and finite
transversal diffusion rate cases coincide and is given by the formula
\[
\dc(u)=\JJ\der{\nu}u(u),
\]
where
\[
\JJ=\int_{\SE_{u}}*du=\int_{\SE_{u}}dv=v_{2}-v_{1}
\]
and
\[
\der{\nu}u(u)=\int_{\SE_{u}}\frac{*du}{|\nabla u|^{2}}=\int_{\SE_{u}}\frac{dv}{|\nabla v|^{2}}.
\]
Observe that we can parametrize a cross section $\SS_{u}$ with a
curve $x:[t_{1},t_{2}]\rightarrow\CH$ with
\[
\dot{x}(t)=\nabla v(x(t)),
\]
so that
\[
\AA(u)=\int_{t_{1}}^{t_{2}}|\dot{x}(t)|=\int_{t_{1}}^{t_{2}}\frac{|\nabla v(x(t))|^{2}}{|\nabla v(x(t)|}.
\]
Hence
\[
\AA(u)=\int_{\SS_{u}}\frac{dv}{|\nabla v|}.
\]

\subsection{Parametric channels}

In this section we will assume that the channel $\CH\subset M$ can
be parametrized by a map 
\[
\varphi:[a,b]\times\Omega\rightarrow M,
\]
where $\Omega$ is a $(n-1)$-dimensional sub-manifold with boundary
of $\RR^{n-1}$. In local coordinates we will write the elements of
$[a,b]\times\Omega$ as $(u,v)$ for $u\in[a,b]$ and $v=(v_{1},\ldots,v_{n-1})\in\RR^{n-1}$.
If denote the of points in $\CH$ by $x$ then we have that $x=\varphi(u,v)$,
which we will simply write as $x=x(u,v)$. We will let the generating
vector field for $\CH$ be
\[
U=\varphi_{*}\left(\pdo u\right),
\]
which has $u$ as a projection function. To compute the effective
diffusion coefficient for $\CH$ (in both the finite and infinite
transversal diffusion rate cases) we will need to compute
\[
\der{\nu}u,\fl{\nabla u},\rho,dh(U)\text{ and }\dv U,
\]
where $h$ is a natural projection function for $\CH$ and $\rho$
its corresponding effective density function. To compute the above
quantities in $(u,v)$-coordinates we will make use of the metric
tensor $g=\varphi^{*}(g_{M})$, where $g_{M}$ is the metric in $M$.
We have that
\[
g=\left(\begin{array}{cc}
\text{\ensuremath{\pd xu\text{\ensuremath{\cdot\pd xu}}}} & \pd xu\cdot\pd xv\\
\left(\pd xu\cdot\pd xv\right)^{T} & g_{v}
\end{array}\right),
\]
where
\[
\pd xu\cdot\pd xv=\left(\pd xu\cdot\pd x{v_{1}},\ldots,\pd xu\cdot\pd x{v_{n-1}}\right)
\]
and $g_{v}$ is the matrix with entries
\[
(g_{v})_{i,j}=\pd x{v_{i}}\cdot\pd x{v_{j}}.
\]
The volume form in $\CH$ is given by 
\[
\mu=\det(g)^{1/2}du\wedge dv,
\]
and hence
\begin{align*}
\der{\nu}u(u) & =\int_{\Omega}\iota_{\pdo u}(\mu)=\int_{\Omega}\det(g(u,v))^{1/2}dv\\
\AA(u) & =\int_{\Omega}\det(g_{v}(u,v))^{1/2}dv
\end{align*}
Observe that
\[
\nabla u=a_{0}\pdo u+\sum_{i=1}^{n-1}a_{i}\pdo{v_{i}}
\]
where
\[
\left(\begin{array}{c}
a_{0}\\
a_{1}\\
\vdots\\
a_{n-1}
\end{array}\right)=g^{-1}\left(\begin{array}{c}
1\\
0\\
\vdots\\
0
\end{array}\right).
\]
Since
\[
a_{0}=\frac{\det(g_{v})}{\det(g)},
\]
we conclude that
\begin{align}
\fl{\nabla u}(u) & =\int_{\Omega}\det(g(u,v))^{\frac{1}{2}}\iota_{\nabla u}(du\wedge dv)\nonumber \\
 & =\int_{\Omega}\left(\frac{\det(g_{v}(u,v))}{\det(g(u,v))^{\frac{1}{2}}}\right)dv.\label{eq:metricJ}
\end{align}
The divergence of $U$ can be computed using the the formula $d(\iota_{U}(\mu))=(\dv U)\mu$.
In our case we have that
\[
d(\iota_{U}(\mu))=d(\det(g)^{1/2}dv)=\pd{\det(g)^{1/2}}udu\wedge dv,
\]
and hence
\[
\dv U=\frac{\pdo u\left(\det(g)^{1/2}\right)}{\det(g)^{1/2}}=\frac{1}{2}\pdo u\left(\log(\det(g))\right).
\]
If $h$ is the natural projection map on the channel then
\[
dh(U)=\pd hu
\]
and
\[
\rho=\left(\int_{\Omega}h(u,v)\det(g(u,v))^{1/2}dv\right)/\left(\int_{\Omega}\det(g(u,v))^{1/2}dv\right)
\]

\bibliographystyle{plain}
\bibliography{myBib}

\end{document}